\documentclass[11pt]{article}

\usepackage[english]{babel}
\usepackage[ansinew]{inputenc}
\usepackage{epsfig}
\usepackage{amssymb}
\usepackage{amsmath}
\usepackage{amsthm}
\usepackage{amsfonts}
\usepackage{color}
\usepackage{subfig}
\usepackage{pseudocode}
\usepackage{hyperref}
\usepackage{graphicx}
\usepackage{fullpage}

\newtheorem{theorem}{Theorem}
\newtheorem{lemma}[theorem]{Lemma}

\newtheorem{definition}[theorem]{Definition}

\newtheorem{proposition}[theorem]{Proposition}

\def\nr{\mathbf{r}}
\def\nb{\mathbf{b}}

\title{On balanced 4-holes in bichromatic point sets
\thanks{This is an arxiv version of \cite{bdfprauv-bh-15}.}}

\author{
S. Bereg\thanks{Department of Computer Science, University of Texas at Dallas, USA. Email: besp@utdallas.edu.
Partially supported by project MEC MTM2009-08652.}
\and 
J. M. D\'{\i}az-B\'{a}\~{n}ez\thanks{Departamento Matem\'atica Aplicada II, Universidad de Sevilla, Spain. Email: \{dbanez,iventura\}@us.es. Partially supported by project MEC MTM2009-08652
and ESF EUROCORES programme EuroGIGA-ComPoSe IP04-MICINN Project EUI-EURC-2011-4306.}
\and 
R. Fabila-Monroy\thanks{Departamento de Matem\'aticas, Cinvestav, Distrito Federal, M\'exico. Email: ruyfabila@math.cinvestav.edu.mx. Partially supported by grant 153984 (CONACyT, Mexico).}
\and 
P. P\'{e}rez-Lantero\thanks{Departamento de Matem\'atica y Ciencia de la Computaci\'on,
Universidad de Santiago, Santiago, Chile. Email: pablo.perez.l@usach.cl.
Partially supported by grant CONICYT, FONDECYT/Iniciaci\'on 11110069 (Chile) 
and project MEC MTM2009-08652.}
\and 
A. Ram\'{\i}rez-Vigueras\thanks{Instituto de Matem\'aticas, UNAM, Mexico. Email: adriana.rv@im.unam.mx. 
Partially supported by CONACyT, Mexico.}
\and 
T. Sakai\thanks{Research Institute of Educational Development, Tokai University, Japan.
Email: sakai@tokai-u.jp. Supported by JSPS KAKENHI Grant Number 24540144.}
\and 
J. Urrutia\thanks{Instituto de Matem\'aticas, UNAM, Mexico. Email: urrutia@matem.unam.mx. 
Partially supported by project MEC MTM2009-08652.}
\and 
I. Ventura\footnotemark[2]
}

\begin{document}
\maketitle

\begin{abstract}
Let $S=R\cup B$ be a point set in the plane in general position
such that each of its elements is colored either red or blue, where
$R$ and $B$ denote the points colored red and the points colored blue, 
respectively.
A quadrilateral with vertices in $S$ is called a $4$-hole if its interior is
empty of elements of $S$.
We say that a $4$-hole of $S$ is balanced if it has $2$ red and $2$ blue points of
$S$ as vertices.
In this paper, we prove that if $R$ and $B$ contain $n$ points each then
$S$ has at least $\frac{n^2-4n}{12}$ balanced $4$-holes, and this bound is tight
up to a constant factor.
Since there are two-colored point sets with no balanced {\em convex}
$4$-holes, we
further provide a characterization of the two-colored point sets 
having this type of $4$-holes. 
\end{abstract}

\section{Introduction}\label{sec:intro}

Let $S$ be a set of points in the plane in general position.
A \emph{hole} of $S$ is a simple polygon $Q$ with vertices in $S$
and with no element of $S$ in its interior. If $Q$ has
$k$ vertices, it is called a $k$-{\em hole} of $P$.
Note that we allow for a $k$-hole to be non-convex.
We will refer to a hole that is not necessarily 
convex as {\em general hole}, and to a hole that is convex
as {\em convex hole}.
%
The study of convex $k$-holes in point sets has
been an active area of research since 
Erd{\H o}s and Szekeres~\cite{Erdos1,Erdos2} asked about the existence of 
$k$ points in convex position in planar point sets. 
It is known that any point set with at least ten
points contains convex $5$-holes~\cite{HARB}.
Horton~\cite{Hort} proved that for 
$k \geq 7$ there are point sets containing no convex $k$-holes.
The question of the existence of convex $6$-holes remained
open for many years, but recently Nicol\'as~\cite{CNico} proved that any point set
with sufficiently many points contains a convex $6$-hole.
A second proof of this result was subsequently given by Gerken~\cite{GERK}.

Recently, the study of general holes of colored point sets has been
started \cite{AichUrr,aichholzer2010}. Let $S=R \cup B$ be a finite set of points in 
general position in the plane. The elements of $R$ and $B$
will be called, respectively, the \emph{red} and \emph{blue} elements
of $S$, and $S$ will be called a {\em bicolored} point set.  
A $4$-hole of $S$ is {\em balanced} if it has two blue and two red vertices. 

In this paper, we address the following question:
Is it true that any bicolored point set with at least two red and two
blue points always has a balanced $4$-hole?
We answer this question in the positive by showing 
that any bicolored point set $S=R\cup B$ with $|R|=|B|\geq 2$
always has a quadratic number of balanced $4$-holes.
We further characterize bicolored point sets that have
balanced convex $4$-holes.

The study of convex $k$-holes in colored point sets 
was introduced  by Devillers et al.~\cite{devillers}. 
They obtained a bichromatic point set with $18$ points 
that contains no convex monochromatic $4$-hole. 
Huemer and Seara~\cite{clemens} obtained
a bichromatic point set with $36$ points containing no
monochromatic $4$-holes. Later, Koshelev~\cite{koshelev}
obtained another such point set with $46$ elements.
Devillers et al.~\cite{devillers} also proved that  every $2$-colored 
Horton set with at least $64$ elements contains an empty 
monochromatic convex $4$-hole. In the same paper the following conjecture
is posed:  Every sufficiently 
large bichromatic point set contains a monochromatic convex 
$4$-hole. This conjecture remains open, and on the other hand
Aichholzer et al~\cite{aichholzer2010} have proved that any bicolored point 
set always has a monochromatic general $4$-hole.
Recently, a result well related with balanced $4$-holes was proved by
Aichholzer et al~\cite{AichUrrVirg2013}: Every two-colored
linearly-separable point set $S=R\cup B$ with $|R|=|B|=n$ contains at least 
$\frac{1}{15}n^2-\theta(n)$ balanced general $6$-holes. 
In a forthcoming paper, the same authors proved the
lower bound $\frac{1}{45}n^2-\theta(n)$ on such holes in the case
where $R$ and $B$ are not necessarily linearly-separable.
One can note that a balanced $6$-hole with vertices $V$ (even if 
$R$ and $B$ are linearly separable)
does not always imply
a balanced $4$-hole with vertices $V'\subset V$ (see, e.g., Figure~\ref{img:6-hole-not-4-hole}).

\begin{figure}[h]
	\centering
	\includegraphics[scale=0.5]{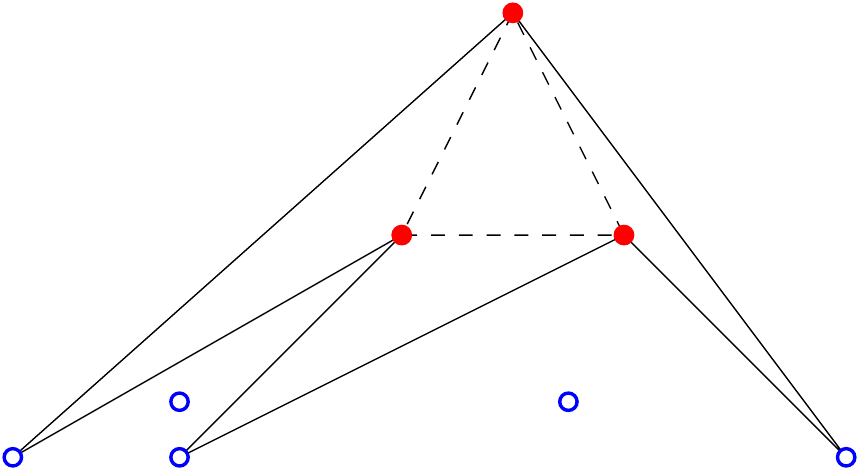}
	\caption{\small{A balanced $6$-hole such that
	no quadruple of its points defines a balanced $4$-hole.
	In the whole paper, red points are represented as solid dots and blue points as tiny circles.}}
	\label{img:6-hole-not-4-hole}
\end{figure}

\medskip

\noindent {\bf Our results:} 
For balanced general $4$-holes, that is, balanced $4$-holes not necessarily
convex, we first show that every bicolored point set 
$S=R\cup B$ with $|R|,|B|\ge 2$ has at least one balanced $4$-hole.
We then prove that if $|R|=|B|=n$ then 
$S$ has at least $\frac{n^2-4n}{12}$ balanced $4$-holes (Theorem~\ref{theo:quadratic-4-holes} of
Section~\ref{sec:lower-bound}), 
and show that this bound is tight up to a constant factor.
This lower bound is improved to $\frac{2n^2+3n-8}{12}$ in the 
case where $R$ and $B$ are linearly separable
(Theorem~\ref{thm:separable-4-holes} of Section~\ref{sec:lower-bound-separable}).
On the other hand, for balanced convex $4$-holes, we
provide a characterization of the bicolored point sets 
$S=R\cup B$ having at least one such hole
(Theorem~\ref{thm:convex-4-holes-not-separable} of Section~\ref{sec:convex-4-holes:no-separability},
and Theorem~\ref{thm:convex-4-hole-separable} of Section~\ref{sec:convex-4-holes:separability}).
Finally, in Section~\ref{sec:discussion}, we discuss extensions of our results such as generalizing the 
above lower bounds
for point sets in which $|R|\neq|B|$, proving the existence of convex $4$-holes either balanced or monochromatic,
deciding the existence of balanced convex $4$-holes, and others.

\medskip

\noindent {\bf General definitions:}
Given any two points $x,y$ of the plane, we denote
by $\overline{xy}$ the straight segment connecting $x$ and $y$,
by $\ell(x,y)$ the line passing through $x$ and $y$,
and by $x\rightarrow y$ the ray that emanates from $x$ and contains $y$.
For every three points $x,y,z$ of the plane, we denote by $\Delta xyz$
the open triangle with vertex set $\{x,y,z\}$. Given $X\subseteq S$,
let $CH(X)$ denote the convex hull of $X$.

Given three non-collinear points $a$, $b$, and $c$, we denote by $\mathcal W(a,b,c)$ the
open convex region bounded by the rays $a \rightarrow b$ and $a \rightarrow c$. 
Given a set $X\subset S$, let $f(a,b,c,X)$ denote a point 
$x\in (X \cap \Delta abc)\cup \{c\}$ minimizing the area of
$\Delta abx$ over all points of $(X \cap \Delta abc)\cup \{c\}$. 

%
 
\section{Lower bounds for general balanced $4$-holes}\label{sec:lower-bound}

It is not hard to see that if $|R|,|B|\ge 2$, then $S$ contains a 
balanced $4$-hole. To prove this,
observe that for every set $H$ of four points there always 
exists a simple polygon
whose vertices are the elements of $H$.
Let $S'$ be a subset of $S$ containing exactly two red points and two blue points,
such that the area of the convex hull of $S'$ is minimum.
Clearly, any simple polygon whose vertex set is $S'$
contains no element of $S$ in its interior, and thus it is
a balanced $4$-hole of $S$. 

On the other hand, if $S$ has
exactly two points of one color and many points of the other color, then
$S$ might contain only a constant number of balanced $4$-holes.
For example, the reader may verify that the point set of
Figure~\ref{img:few-holes} contains exactly five balanced $4$-holes.

\begin{figure}[h]
	\centering
	\includegraphics[scale=0.5]{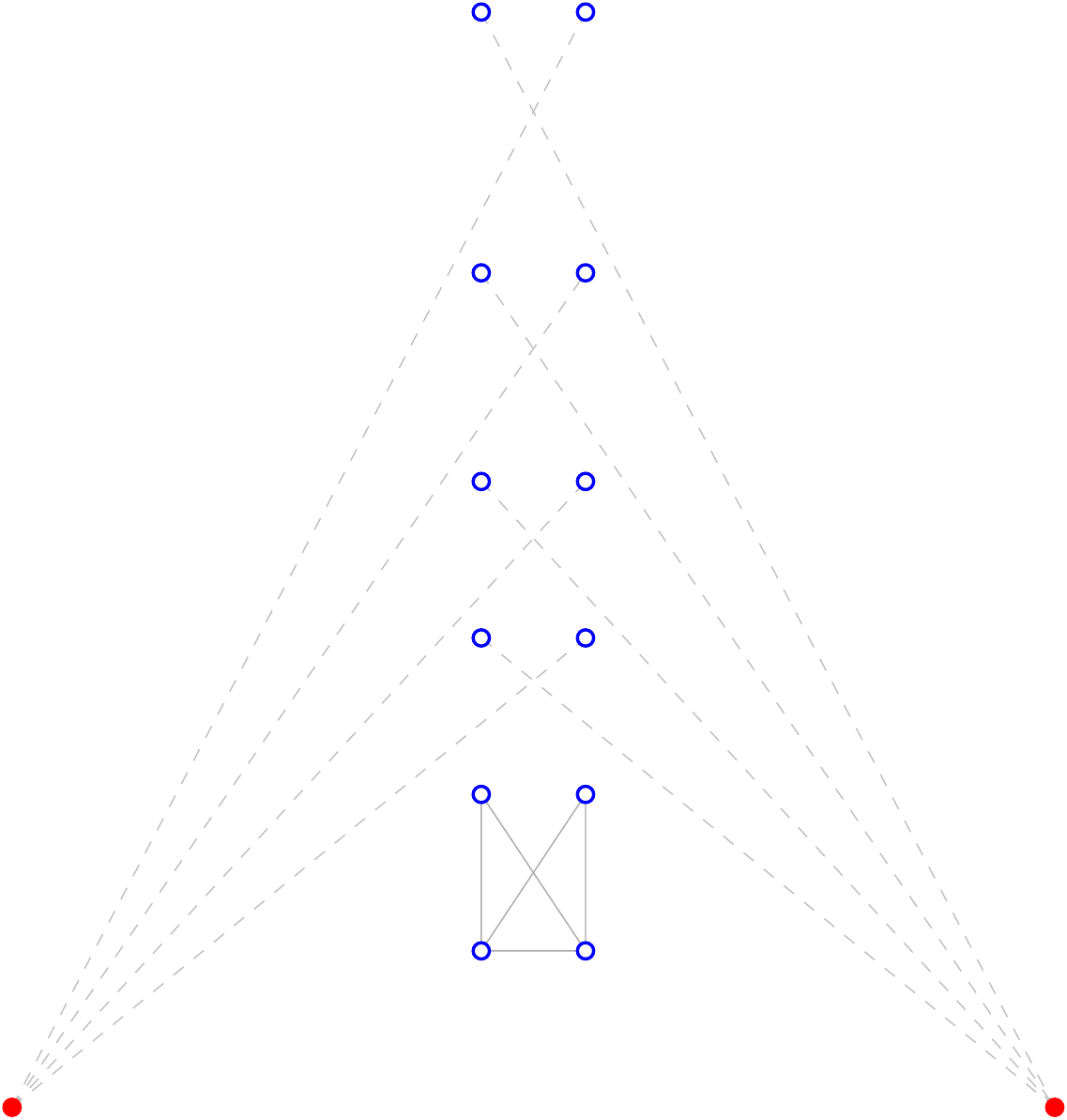}
	\caption{\small{A point set with exactly
	five balanced $4$-holes, obtained by choosing
	the two red points and any pair of blue points connected by a continuous segment.}}
	\label{img:few-holes}
\end{figure}

In the case where $|R|=|B|=n$, $S$ has (at least) a linear
number of balanced $4$-holes. 
Indeed, by applying the
ham-sandwich theorem recursively, we can partition
$S$ into a linear number of constant size disjoint subsets whose convex 
hulls are pairwise disjoint, and each of them contains
at least two red points and two blue points, and has thus a $4$-hole.

In this section we prove the following stronger result:

\begin{theorem}\label{theo:quadratic-4-holes}
Let $S=R \cup B$ be a set of $2n$ points in general position in the plane
such that $|R|=|B|=n$. Then $S$
has at least $\frac{n^2-4n}{12}$ balanced $4$-holes.
\end{theorem}

We consider some definitions and preliminary results to 
prove Theorem~\ref{theo:quadratic-4-holes}.
In the rest of this section we will assume that $|R|=|B|=n$.

Given two points $p,q\in S$ with different colors, let $T(p,q)$ be the set of the at most 
four points obtained by taking the first point found in each of the next
four rotations: 
the rotation of $p\rightarrow q$ around $p$ clockwise;
the rotation of $p\rightarrow q$ around $p$ counter-clockwise;
the rotation of $q\rightarrow p$ around $q$ clockwise; and
the rotation of $q\rightarrow p$ around $q$ counter-clockwise.

We classify (or color) the edge $\overline{pq}$
with one of the following four colors: {\em green}, {\em black}, {\em red}, and {\em blue}. 
We color $\overline{pq}$ green if it is an edge, or a diagonal,
of some balanced $4$-hole. 
If $\overline{pq}$ is an
edge of the convex hull of $S$ and is not green, then $\overline{pq}$ is colored black.
If $\overline{pq}$ is neither green nor black,
then all the points in $T(p,q)$ must have the same color and there
are elements of $T(p,q)$ to each side of $\ell(p,q)$. We then
color $\overline{pq}$ with the color of the points in $T(p,q)$.

\begin{figure}[h]
	\centering
	\subfloat[]{
		\includegraphics[scale=0.6,page=1]{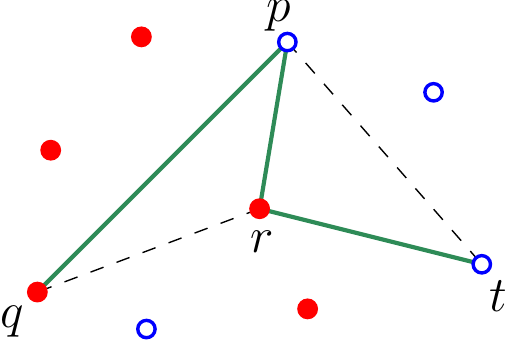}
		\label{fig:green-edge}
	}\hspace{0.5cm}
	\subfloat[]{
		\includegraphics[scale=0.6,page=2]{img/edge-color.pdf}
		\label{fig:black-edge}
	}\hspace{0.5cm}
	\subfloat[]{
		\includegraphics[scale=0.6,page=3]{img/edge-color.pdf}
		\label{fig:red-edge}
	}\hspace{0.5cm}
	\caption{\small{The edge colors: (a) The polygon with vertex set $\{p,q,r,t\}$
	is a balanced $4$-hole, then the edges $\overline{pq}$, $\overline{pr}$, and $\overline{rt}$
	are colored green. (b) Since the edge $\overline{pq}$ is a convex hull edge and there
	is no balanced $4$-hole with edge $\overline{pq}$, then $\overline{pq}$ is colored black.
	(c) Since $\overline{pq}$ is neither red nor black, and the elements of $T(p,q)$ are red,
	$\overline{pq}$ is colored red.}}
	\label{fig:edge-colors}
\end{figure}

\begin{lemma}\label{lem:number-red-blue-edges}
The number of red edges and the number of blue edges are
each at most $n\lfloor\frac{n-1}{3}\rfloor$.
\end{lemma}

\begin{proof}
Let $r\in R$ be any red point. Sort the elements $B$ radially around $r$
in counter-clockwise order, and label them $b_0,b_1,\ldots,b_{n-1}$ in this order.
Subindices are taken modulo $n$.

Suppose that the edge $\overline{rb_i}$ is red, $0\le i<n$,
and the angle needed to rotate the ray $r\rightarrow b_i$
counter-clockwise around $r$ in order to reach $r\rightarrow b_{i+1}$ is less than $\pi$. 
If $\Delta rb_ib_{i+1}$ does not contain
elements of $R$, then there must exist a red point $z$ in 
$\mathcal{W}(rb_ib_{i+1})\setminus \Delta rb_ib_{i+1}$. Then,
the quadrilateral with vertex set $\{r,b_i,z',b_{i+1}\}$ is a balanced $4$-hole, where
$z':=f(b_i,b_{i+1},z,R)$, which contradicts that $\overline{rb_i}$ is red 
(see Figure~\ref{img:3-red-points-at-least-0}). Hence,
$\Delta rb_ib_{i+1}$ must contain red points. 
In fact, $\Delta rb_ib_{i+1}$ contains
at least three red points in order to avoid that $r$, $b_i$, and $b_{i+1}$, joint
with some red point in $\Delta rb_ib_{i+1}$, form a balanced $4$-hole with edge $\overline{rb_i}$
(see Figure~\ref{img:3-red-points-at-least-a} and Figure~\ref{img:3-red-points-at-least-b}).
%
By symmetry, for every red edge $\overline{rb_i}$, if $\overline{rb_{i-1}}$
is reached by rotating $\overline{rb_i}$ clockwise by an angle less than $\pi$, then
the triangle $\Delta rb_{i-1}b_i$ contains at least three red points.
These observations
imply that the number of red edges among 
$\overline{rb_0},\overline{rb_1},\ldots,\overline{rb_{n-1}}$ (i.e.\ the
number of red edges incident to $r$) is at most
$\lfloor\frac{n-1}{3}\rfloor$. Summing over all the red points,
the total number of red edges is at most $n\lfloor\frac{n-1}{3}\rfloor$. 

\begin{figure}[h]
	\centering
	\subfloat[]{
		\includegraphics[scale=0.7,page=4]{img/3-red-points-at-least.pdf}
		\label{img:3-red-points-at-least-0}
	}\hspace{0.5cm}
	\subfloat[]{
		\includegraphics[scale=0.7,page=1]{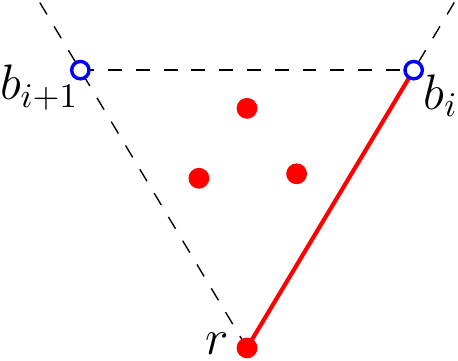}
		\label{img:3-red-points-at-least-a}
	}\hspace{0.5cm}
	\subfloat[]{
		\includegraphics[scale=0.7,page=2]{img/3-red-points-at-least.pdf}
		\label{img:3-red-points-at-least-b}
	}
	\caption{\small{(a) If $\mathcal{W}(rb_ib_{i+1})$ contains red points
	and $\Delta rb_ib_{i+1}$ does not, then there exists a balanced $4$-hole with edge $\overline{rb_i}$.
	(b) If the edge $\overline{rb_i}$ is red then the triangle 
	$\Delta rb_ib_{i+1}$ must contain at least three points in order
	to block balanced $4$-holes with vertices $r$, $b_i$, $b_{i+1}$, 
	and some red point of $\Delta rb_ib_{i+1}$, having $\overline{rb_i}$ as edge.
	(c) If $\Delta rb_ib_{i+1}$ contains exactly one or two red points
	then there is a balanced $4$-hole with edge $\overline{rb_i}$.}}
	\label{img:3-red-points-at-least}
\end{figure}

Analogously, 
the total number of blue edges is also at most $n\lfloor\frac{n-1}{3}\rfloor$. 
\end{proof}

\begin{lemma}\label{lem:number-green-edges}
The number of green edges is at least $\frac{n^2-4n}{3}$.
\end{lemma}

\begin{proof}
There are $n^2$ bichromatic edges in total. By Lemma~\ref{lem:number-red-blue-edges},
at most $n\lfloor\frac{n-1}{3}\rfloor$ of them are red and at most $n\lfloor\frac{n-1}{3}\rfloor$ are blue.
Further observe that at most $2n$ edges are black. Then the number of green edges is at least:
$$n^2-2n\left\lfloor\frac{n-1}{3}\right\rfloor-2n\geq \frac{n^2-4n}{3}.$$
\end{proof}

Observe now that any balanced $4$-hole defines at
most four green edges as polygonal edges or diagonals. Thus,
by Lemma~\ref{lem:number-green-edges},
the number of balanced general
$4$-holes is at least $\frac{1}{4}\left( \frac{n^2-4n}{3}\right)= \frac{n^2-4n}{12}$,
and Theorem~\ref{theo:quadratic-4-holes} thus follows.

\subsection{The separable case}\label{sec:lower-bound-separable}

We now improve our bounds of the previous section
for the case where $R$ and $B$ are linearly separable.
Suppose without loss of generality that there is a horizontal line
$\ell$ such that the elements in
$R$ are above $\ell$, and those in
$B$ are below $\ell$. Further assume that
no two elements in $S=R\cup B$ have the
same $y$-coordinate. 

\begin{lemma}\label{lem:separable:number-red-blue-edges}
If $R$ and $B$ are linearly separable then both the number
of red edges and the number of blue edges are each at most $\frac{n^2-3n+2}{6}$.
\end{lemma}

\begin{proof}
Label the red points $r_0,r_1,\ldots,r_{n-1}$ in the ascending order of the $y$-coordinates.
Let $r_i$ be any red point, $0\leq i<n$. 
Sort the blue points radially around $r_i$
in counter-clockwise order and label them $b_0,b_1,\ldots,b_{n-1}$ in this order.
Similarly as in the proof of Lemma~\ref{lem:number-red-blue-edges}, 
if $\overline{r_ib_j}$ is red, $0\le j< n$, then among $r_0,r_1,\ldots,r_{i-1}$
the triangle $\Delta r_ib_jb_{j-1}$ contains at least three elements if $j>0$, 
and the triangle $\Delta r_ib_jb_{j+1}$ contains at least three elements if $j<n-1$.
%
%
Then the number of red edges incident to $r_i$ is at most
$\lfloor\frac{i}{3}\rfloor$, and over all the red points, the number
of red edges is at most
$$\sum_{i=0}^{n-1}\left\lfloor\frac{i}{3}\right\rfloor$$
If $n-1=3k$, for some integer $k$, then:
\begin{equation*}
\sum_{i=0}^{n-1}\left\lfloor\frac{i}{3}\right\rfloor =3\left(0+1+\ldots+(k-1)\right)+k
= \frac{n^2-3n+2}{6}.
\end{equation*}

If $n-1=3k+1$, then:
\begin{equation*}
\sum_{i=0}^{n-1}\left\lfloor\frac{i}{3}\right\rfloor = 3\left(0+1+\ldots+(k-1)\right)+2k
= \frac{n^2-3n+2}{6}.
\end{equation*}

Finally, if $n-1=3k+2$, then:
\begin{equation*}
\sum_{i=0}^{n-1}\left\lfloor\frac{i}{3}\right\rfloor  =  3\left(0+1+\ldots+k\right)
= \frac{n^2-3n}{6}.
\end{equation*}

Therefore, we have that the number of red edges is at most $\frac{n^2-3n+2}{6}$.
Analogously, there are at most $\frac{n^2-3n+2}{6}$ blue
edges in total.
\end{proof}

\begin{theorem}\label{thm:separable-4-holes}
If $R$ and $B$ are linearly separable then 
the number of balanced $4$-holes is at least $\frac{2n^2+3n-8}{12}$.
\end{theorem}

\begin{proof}
Since $R$ and $B$ are linear separable, the number of black edges
is at most $2$. Using Lemma~\ref{lem:separable:number-red-blue-edges},
we can ensure that the number of green edges is at least
\begin{eqnarray*}
n^2-2\left(\frac{n^2-3n+2}{6}\right)-2 & = & \frac{2n^2+3n-8}{3}.
\end{eqnarray*}
Then the number of balanced $4$-holes is at least $\frac{2n^2+3n-8}{12}=\frac{2n^2+3n-8}{12}$.
\end{proof}

We observe that our lower bounds are asymptotically tight for point sets $S=R\cup B$ with
$|R|=|B|=n$. For example, if $R$ and $B$ are far enough from each other (i.e.\ 
any line passing through two points of $R$
does not intersect $CH(B)$, and vice versa), $R$ is a concave chain,
and $B$ a convex chain (see Figure~\ref{img:example-quadratic-4-holes}),
then the number of balanced $4$-holes is precisely $(n-1)\times (n-1)$;
each of them convex and formed by two consecutive red points and two consecutive blue points.
This point set $R\cup B$ (without the colors) was called the {\em double chain}~\cite{garciaNT00}. 

\begin{figure}[h]
	\centering
	\includegraphics[scale=0.5]{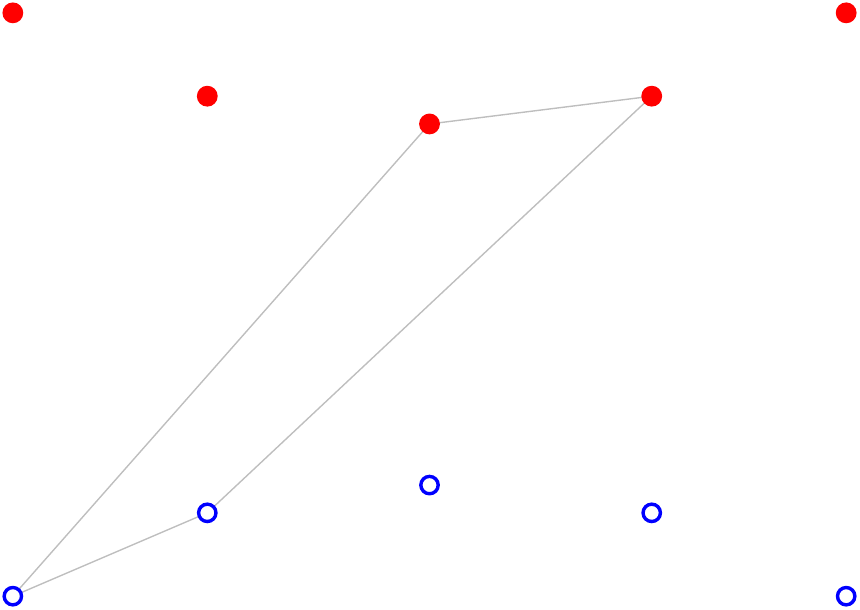}
	\caption{\small{An example of $2n$ points having exactly $(n-1)^2$ balanced $4$-holes.}}
	\label{img:example-quadratic-4-holes}
\end{figure}

\section{Balanced convex $4$-holes}\label{sec:characterization}

In this section we characterize bicolored point sets $S=R \cup B$
that contain balanced \emph{convex} $4$-holes.
To start with, we point out that in general 
$S=R \cup B$ does not necessarily have balanced convex $4$-holes.
The point sets shown in Figure~\ref{fig:no-convex-4-hole}
do not have balanced convex $4$-holes.
Observe that the number of blue points in the interior of the convex hull of the blue points
in Figure~\ref{fig:no-convex-4-hole-b}
and Figure~\ref{fig:no-convex-4-hole-c} can be arbitrarily large.
A more general example with eight points, 4 red and 4 blue
linearly separable, is shown in
Figure~\ref{fig:no-convex-4-hole-d}, which can
be generalized to point sets with $2n$ points, $n \geq 2$,
$n$ red and $n$ blue. 

\begin{figure}[h]
	\centering
	\subfloat[]{
		\includegraphics[scale=0.5]{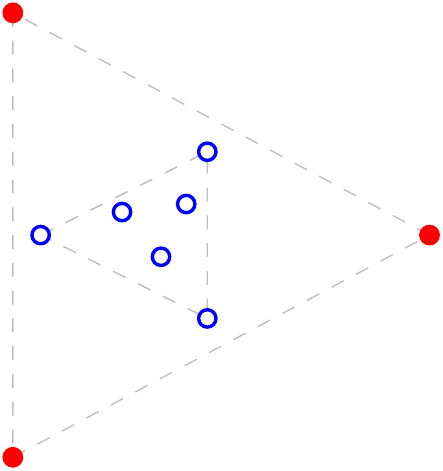}
		\label{fig:no-convex-4-hole-b}
	}\hspace{0.5cm}
	\subfloat[]{
		\includegraphics[scale=0.5]{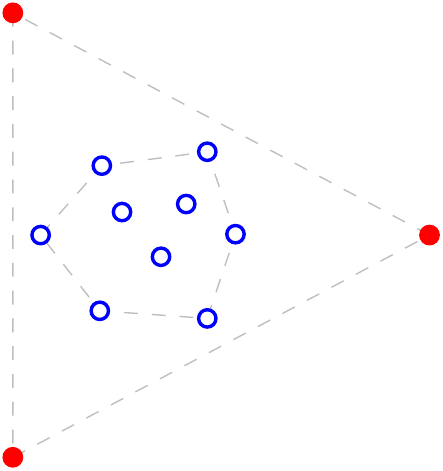}
		\label{fig:no-convex-4-hole-c}
	}\\
	\subfloat[]{
		\includegraphics[scale=0.5]{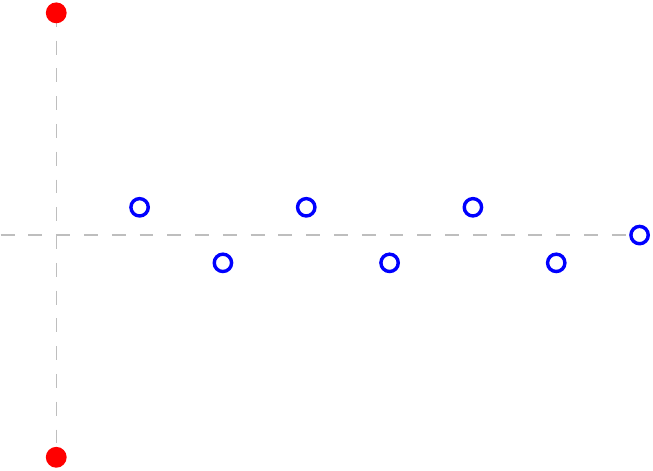}
		\label{fig:no-convex-4-hole-a}
	}\hspace{0.5cm}
	\subfloat[]{
		\includegraphics[scale=0.45]{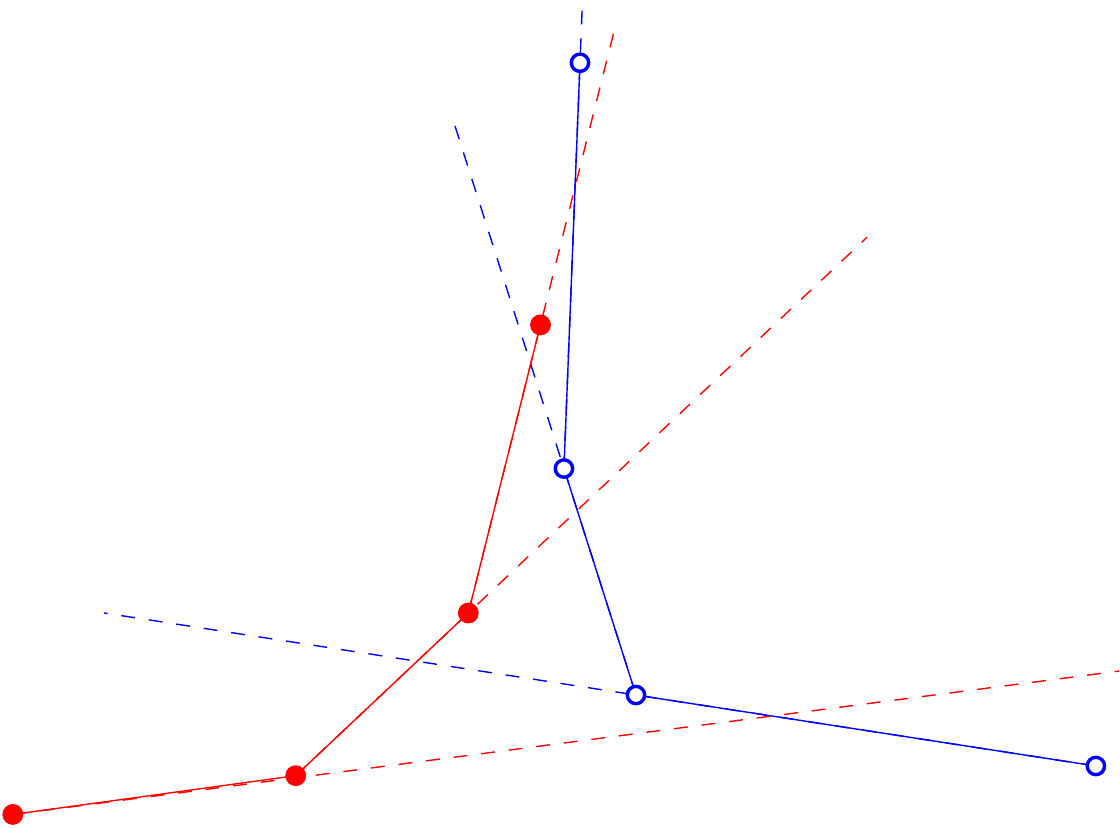}
		\label{fig:no-convex-4-hole-d}
	}
	\caption{\small{Some point sets with no balanced convex $4$-holes.}}
	\label{fig:no-convex-4-hole}
\end{figure}

Let $p,q\in S$ be two points of the same color. If $p$ and $q$ are red, $\overline{pq}$
will be called a \emph{red-red edge}.
Otherwise, if $p$ and $q$ are blue, we call it a \emph{blue-blue edge}.

\subsection{$R$ and $B$ are not linearly separable}\label{sec:convex-4-holes:no-separability}

We proceed now to characterize bicolored point sets $S=R\cup B$, not
linearly separable, which contain
balanced convex $4$-holes. We assume $|R|, |B| \geq 2$.

\begin{lemma}\label{lem:crossing-edges}
If $S$ contains a red-red edge and a blue-blue edge that intersect each other,
then $S$ contains a balanced convex 4-hole.
\end{lemma}

\begin{proof}
Choose a red-red edge $\overline{ab}$ and a blue-blue edge $\overline{cd}$
such that $\overline{ab}\cap \overline{cd}\neq \emptyset$ 
and the convex quadrilateral $Q$ with vertex set $\{a,b,c,d\}$ is 
of minimum area among all possible convex quadrilaterals having a red-red
diagonal and a blue-blue diagonal.
Observe that $Q$ is balanced and assume that $Q$ is not a $4$-hole. 
Then $Q$ contains a point of $S$ in its interior. 
Suppose w.l.o.g.\ that there is a red point
$e$ in the interior of $Q$. Then we have that $\overline{ea}$ 
intersects $\overline{cd}$, or $\overline{eb}$ intersects $\overline{cd}$.
Suppose w.l.o.g.\ the former case. Hence,
$\{a,e,c,d\}$ is the vertex set of 
a balanced convex quadrilateral with a red-red diagonal and a blue-blue diagonal 
with area smaller than that of $Q$,
a contradiction.
\end{proof}

\begin{lemma}\label{lem:crossing-ch-boundaries}
If the boundaries of $CH(R)$ and $CH(B)$ intersect each other,
then $S$ contains a balanced convex 4-hole.
\end{lemma}

\begin{proof}
Observe that there exist a red-red edge and a blue-blue edge that intersect each other.
Therefore, the result follows from Lemma~\ref{lem:crossing-edges}.
\end{proof}

\begin{lemma}\label{lem:ch-contains-other-R=3}
Let $S=R\cup B$ be a bichromatic point set such that $R$ and $B$ are not
linearly separable, $CH(B)\subset CH(R)$, $|R|=3$, and $|B|\ge 2$. Then
$S$ contains a balanced convex $4$-hole if and only if there is a blue-blue
edge $\overline{uv}$ of $CH(B)$ such that one of the open half-planes bounded
by $\ell(u,v)$ contains exactly 2 red points and no blue point.
\end{lemma}

\begin{proof}
Let $a,b,c$ denote the three elements of $R$. Suppose that there exists an edge 
$\overline{uv}$ of $CH(B)$ such that $a$ and $b$ belong to one of the
two open half-planes
bounded by $\ell(u,v)$ and that the elements of $S\setminus \{a,b,u,v\}$ belong to the other open
half-plane (see Figure~\ref{fig:blue-inside-3-red-a}). 
Then the quadrilateral with vertex set $\{a,b,u,v\}$ is a balanced
convex $4$-hole.

\begin{figure}[h]
	\centering
		\includegraphics[scale=0.6]{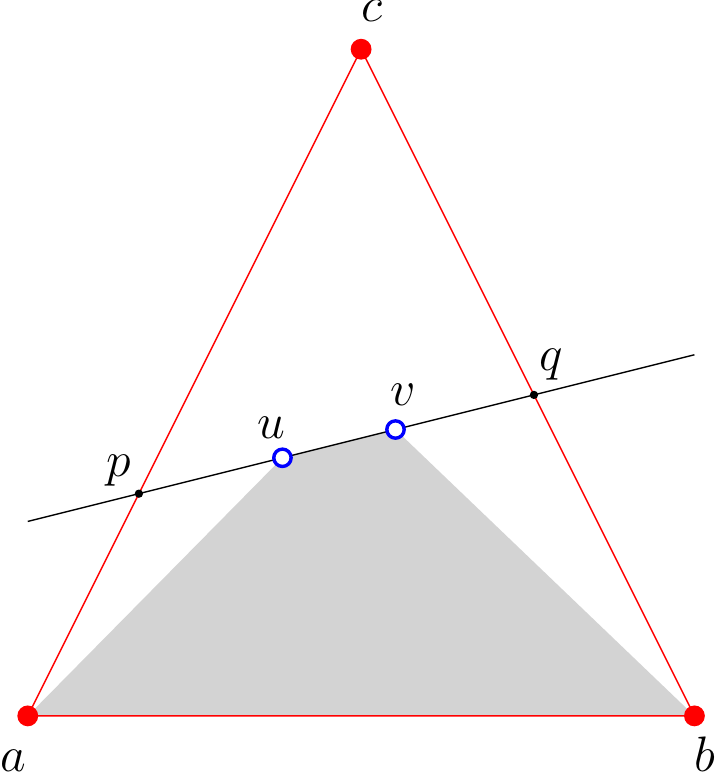}
	\caption{\small{Proof of Lemma~\ref{lem:ch-contains-other-R=3}.}}
	\label{fig:blue-inside-3-red-a}
\end{figure}

Suppose now that $S$ has a balanced convex $4$-hole. Assume w.l.o.g.\ that this
$4$-hole has vertex set $\{a,b,u,v\}$, where $u\rightarrow v$ intersects
$\overline{bc}$ at the point $q$, and $v\rightarrow u$ intersects $ac$ at
the point $p$ (see Figure~\ref{fig:blue-inside-3-red-a}).
We select the points $u$ and $v$ so that the sum of the distance of $u$ from the segment
$\overline{ac}$ and the distance of $v$ from the segment $\overline{bc}$ is minimized.
Observe from the choice of $u$ and $v$ that $\Delta aup\cup \Delta bvq$
does not contain blue points. Hence, the result follows.
%
\end{proof}

\begin{lemma}\label{lem:ch-contains-other-R>=4,B>=2}
Let $S=R\cup B$ be a bicolored point set such that $R$ and $B$ are not
linearly separable, $CH(B)\subset CH(R)$, $|R|\ge 4$, and $|B|\ge 2$.
Then $S$ has a balanced convex $4$-hole.
\end{lemma}

\begin{proof}
Let $\mathcal{T}$ be a triangulation of $R$. If there are two blue points that belong
to different triangles of $\mathcal{T}$, then there exist a red-red edge and a blue-blue
edge intersecting each other, and the result thus follows from Lemma~\ref{lem:crossing-edges}.
Suppose then that $B$ is completely contained in a single triangle $t$ of $\mathcal{T}$, with vertices
$a,b,c\in R$ in counter-clockwise order. 

If $|B|=2$, there exists an edge of $t$ which is not intersected by the line
through the two blue points. Then the two red points of that edge, joint with the two
blue points, form a balanced convex $4$-hole (Lemma~\ref{lem:ch-contains-other-R=3}).

Suppose then that $|B|\ge 3$, thus $CH(B)$ has at least three vertices.
Since $|R|\ge 4$ there exists a triangle $t'$ of 
$\mathcal{T}$ sharing an edge with $t$. Assume w.l.o.g.\ that such an edge is $\overline{ab}$,
and denote by $d$ the other vertex of $t'$. Further assume w.l.o.g.\ that $\ell(a,b)$ is horizontal,
and $d$ is below $\ell(a,b)$.

Let $u:=f(a,b,c,B)$. Observe that $u$ is a vertex of $CH(B)$. Let $v\in B$ denote
the vertex succeeding $u$ in $CH(B)$ in the counter-clockwise order, and
$w\in B$ denote
the vertex succeeding $u$ in $CH(B)$ in the clockwise order. Both $v$ and $w$ are not below
the horizontal line through $u$ by the definition of $u$. 
If either $\ell(u,w)$ or $\ell(u,v)$ does not
intersect $\overline{ab}$, 
then there is a balanced convex $4$-hole by Lemma~\ref{lem:ch-contains-other-R=3}.
Suppose then that both $\ell(u,w)$ and $\ell(u,v)$ intersect
$\overline{ab}$. Refer to Figure~\ref{fig:blue-inside-4-red}.

\begin{figure}[h]
	\centering
	\subfloat[Case 1]{
		\includegraphics[scale=0.65]{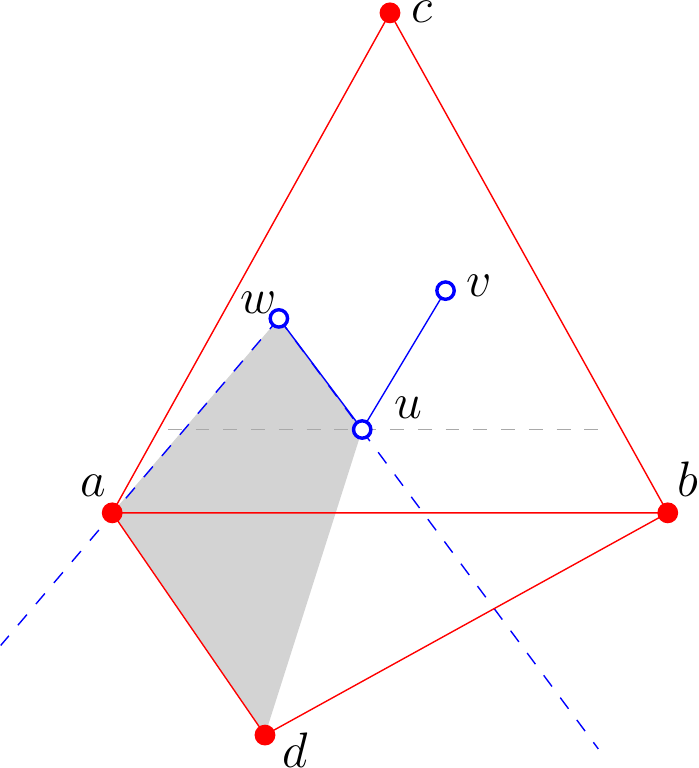}
		\label{fig:blue-inside-4-red-a}
	}\hspace{0.2cm}
	\subfloat[Case 2]{
		\includegraphics[scale=0.65]{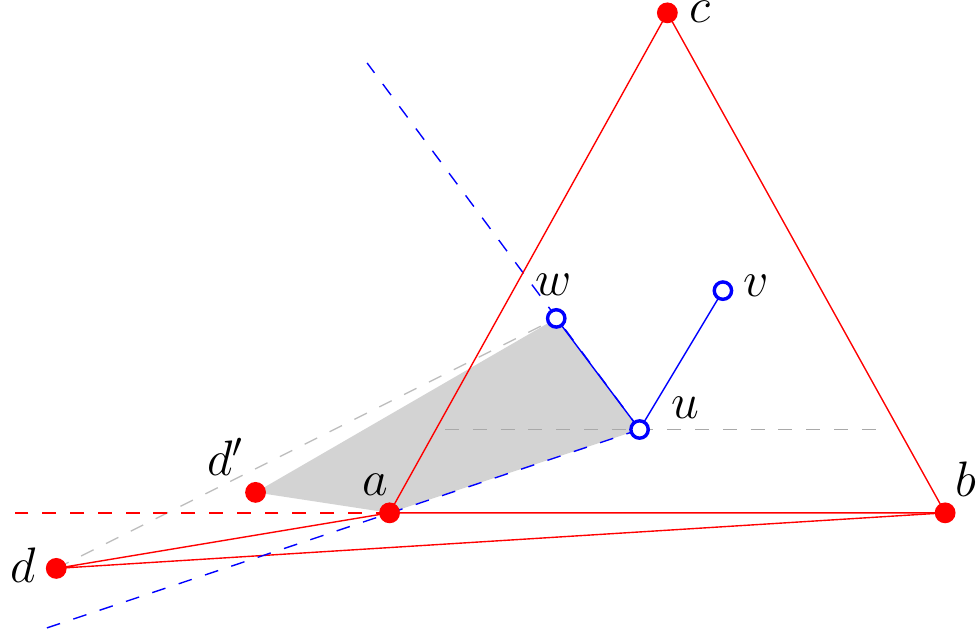}
		\label{fig:blue-inside-4-red-b}
	}\\
	\subfloat[Case 3]{
		\includegraphics[scale=0.65]{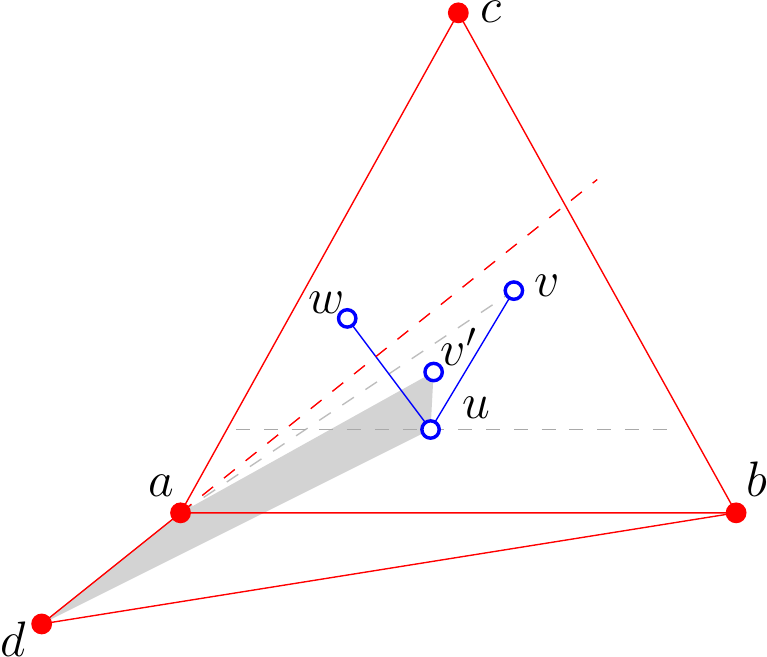}
		\label{fig:blue-inside-4-red-c}
	}\hspace{0.2cm}
	\subfloat[Case 4]{
		\includegraphics[scale=0.65]{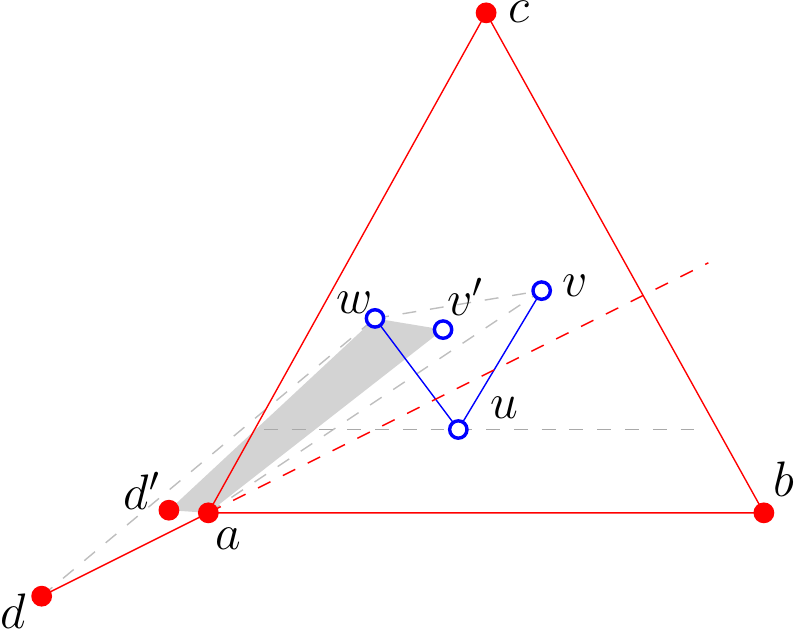}
		\label{fig:blue-inside-4-red-d}
	}
	\caption{\small{Proof of Lemma~\ref{lem:ch-contains-other-R>=4,B>=2}.}}
	\label{fig:blue-inside-4-red}
\end{figure}

We consider the following four cases according to the possible locations of point $d$,
by assuming w.l.o.g.\ that point $d$ is to the left of $\ell(u,w)$. The other symmetric cases 
arise when $d$ is to the right of $\ell(u,v)$.

{\em Case 1:} $d\in\mathcal{W}(w,a,u)$ (see Figure~\ref{fig:blue-inside-4-red-a}).
The quadrilateral with vertex set $\{a,d,u,w\}$ is a balanced convex $4$-hole.

{\em Case 2:} $d\in\mathcal{W}(u,a,w)$ (see Figure~\ref{fig:blue-inside-4-red-b}).
The quadrilateral with vertex set $\{d',a,u,w\}$ is a balanced convex $4$-hole,
where $d'=f(w,a,d,R)$.

{\em Case 3:} $d\notin \mathcal{W}(w,a,u)\cup \mathcal{W}(u,a,w)$ and $\ell(a,d)\cap\overline{uv}=\emptyset$ 
(see Figure~\ref{fig:blue-inside-4-red-c}).
The quadrilateral with vertex set $\{a,d,u,v'\}$ is a balanced convex $4$-hole,
where $v'=f(a,u,v,B)$.

{\em Case 4:} $d\notin \mathcal{W}(w,a,u)\cup \mathcal{W}(u,a,w)$ and $\ell(a,d)\cap\overline{uv}\neq\emptyset$ 
(see Figure~\ref{fig:blue-inside-4-red-d}).
The quadrilateral with vertex set $\{d',a,v',w\}$ is a balanced convex $4$-hole,
where $d'=f(a,w,d,R)$ and $v'=f(a,w,v,B)$.

Since any location of $d$ is covered by one of the above cases 
(or by one of their symmetric ones),
there exists a balanced convex $4$-hole. The result follows.
\end{proof}

By combining Lemma~\ref{lem:crossing-ch-boundaries}, Lemma~\ref{lem:ch-contains-other-R=3}, and
Lemma~\ref{lem:ch-contains-other-R>=4,B>=2}, we obtain the following result that completely
characterizes the non-linearly separable bichromatic point sets that have a balanced convex $4$-hole.

\begin{theorem}\label{thm:convex-4-holes-not-separable}
Let $S=R \cup B$ be a bichromatic point set such that $R$ and $B$ are not linearly separable.
Then $S$ has a balanced convex $4$-hole if and only if one 
of the following conditions holds:
\begin{enumerate}
\item $CH(B)\subset CH(R)$, $|R|=3$, $|B|\ge 2$, and there is a blue-blue
edge $\overline{uv}$ of $CH(B)$ such that one of the open half-planes bounded
by $\ell(u,v)$ contains exactly 2 red points and no blue point.
\item $CH(R)\subset CH(B)$, $|B|=3$, $|R|\ge 2$, and there is a red-red
edge $\overline{uv}$ of $CH(R)$ such that one of the open half-planes bounded
by $\ell(u,v)$ contains exactly 2 blue points and no red point.
\item $CH(B) \subset CH(R)$, $|R| \geq 4$, $|B| \geq 2$,
\item $CH(R) \subset CH(B)$, $|B| \geq 4$, $|R| \geq 2$,
\item The boundaries of $CH(B)$ and $CH(R)$ intersect each other.
\end{enumerate}
\end{theorem}

\subsection{$R$ and $B$ are linearly separable}\label{sec:convex-4-holes:separability}

In the rest of this section, we will assume that $R$ and $B$ are linearly separable. 
At first glance, one might be tempted to think that if the cardinalities
of $R$ and $B$ are large enough, 
then $S$ always contains balanced convex $4$-holes.
This certainly happens in the point set of Figure~\ref{img:example-quadratic-4-holes},
in which $R$ and $B$ are far enough from each other.
There are, however, examples of linearly separable bicolored point sets
with an arbitrarily large number of points that do not contain any
balanced convex $4$-hole. For instance, the point set shown in 
Figure~\ref{fig:no-convex-4-hole-d} has no balanced convex $4$-hole. 
Observe in this example that if we choose a red-red edge and a blue-blue edge,
the convex hull of their vertices is either a triangle or a convex quadrilateral
that contains at least one other point in its interior. 

Given an edge $e$ of $CH(R)$ and an edge $e'$ of $CH(B)$,
we say that $e$ and $e'$ {\em see each other} if the union
of the sets of their vertices defines a balanced convex $4$-hole whose
interior intersects with neither $CH(R)$ nor $CH(B)$. We 
assume that there exists a non-horizontal line $\ell$ such that
the elements of $R$ are located to the left of $\ell$ and
the elements of $B$ are located to the right.

\begin{definition}
Let $S=R \cup B$ be a bicolored point set such that $R$ and $B$ are linearly separable.
Conditions C1 and C2 are defined as follows:
\begin{itemize}
\item[C1.] There exist an edge $e$ of $CH(R)$ and an edge $e'$ of $CH(B)$ such that
$e$ and $e'$ see each other.
\item[C2.] There exists an edge $\overline{uv}$ of $CH(R)$ and points $b,z\in B$ such that
$z\in \Delta uvb$, $R\cap \Delta uvb=\emptyset$, and $R\cap\mathcal{W}(b,u,v)\neq\emptyset$;
or this statement holds if we swap $R$ and $B$.
\end{itemize}
\end{definition}

\begin{lemma}\label{lem:red-blue-edge-crossed}
Let $S=R \cup B$ be a bicolored point set such that $R$ and $B$ are linearly separable.
If there exist a point $r\in R$, a point $b\in B$, an edge $e$ of $CH(R)$, and an
edge $e'$ of $CH(B)$, such that the interiors of $e$ and $e'$ intersect with the interior
of $\overline{rb}$, then C1 or C2 holds.
\end{lemma}

\begin{proof}
Let $u$ and $v$ be the endpoints of $e$ and $w$ and $z$ the endpoints of $e'$.
Assume
w.l.o.g.\ that $\ell(r,b)$ is horizontal, $u$ and $w$ are above $\ell(r,b)$,
and then $v$ and $z$ are below $\ell(r,b)$. If $e$ and $e'$ see each other (see Figure~\ref{fig:linear-sep-lemma-a}),
then {\em C1} holds.
Otherwise, assume w.l.o.g.\ that $z$ is contained in $\Delta uvw$ (see Figure~\ref{fig:linear-sep-lemma-b}).
We have $z\in \Delta uvb$ because $z$ lies between the intersections of $\ell(w,z)$ with
$\overline{rb}$ and $\overline{uv}$, which both are in the closure of $\Delta uvb$.
This implies that $R\cap \Delta uvb=\emptyset$ and $r\in\mathcal{W}(b,u,v)$.
Then {\em C2} is satisfied.
\begin{figure}[h]
	\centering
	\subfloat[]{
		\includegraphics[scale=0.6]{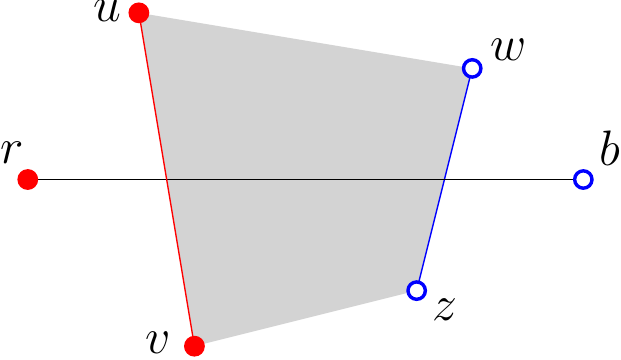}
		\label{fig:linear-sep-lemma-a}
	}\hspace{0.5cm}
	\subfloat[]{
		\includegraphics[scale=0.6]{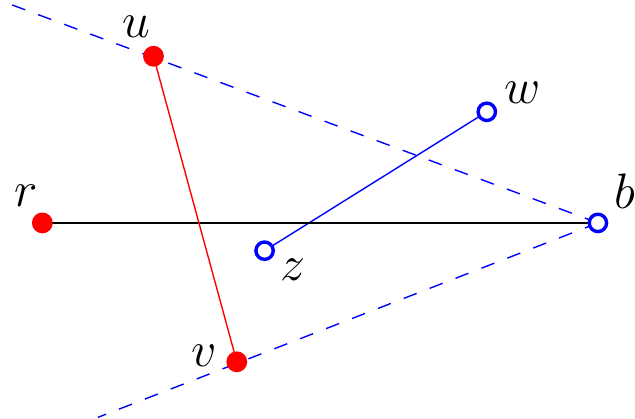}
		\label{fig:linear-sep-lemma-b}
	}
	\caption{\small{Proof of Lemma~\ref{lem:red-blue-edge-crossed}.}}
	\label{fig:linear-sep-lemma}
\end{figure}
\end{proof}

\begin{theorem}\label{thm:convex-4-hole-separable}
A bichromatic point set $S=R \cup B$, such that $R$ and $B$ are linearly separable,
has a balanced convex $4$-hole if and only if C1 or C2 holds.
\end{theorem}

\begin{proof}
If condition {\em C1} holds then $S$ has trivially a balanced convex $4$-hole.
Then suppose that condition {\em C2} holds. 
Let $z':=f(u,v,b,B)$ and observe that $z'\neq b$ since $z\in \Delta uvb$.
Let $r$ be any red point in $R\cap\mathcal{W}(b,u,v)$ (see Figure~\ref{fig:linear-sep-c}).
Observe that we have either $r\in\mathcal{W}(b,u,z')$ or $r\in\mathcal{W}(b,z',v)$.
Assume w.l.o.g.\ the former case. 
Then the quadrilateral with
vertex set $\{r',z',b',u\}$ is a balanced convex $4$-hole,
where $r':=f(u,z',r,R)$ and $b':=f(u,z',b,B)$.  

\begin{figure}[h]
	\centering
	\subfloat[]{
		\includegraphics[scale=0.6]{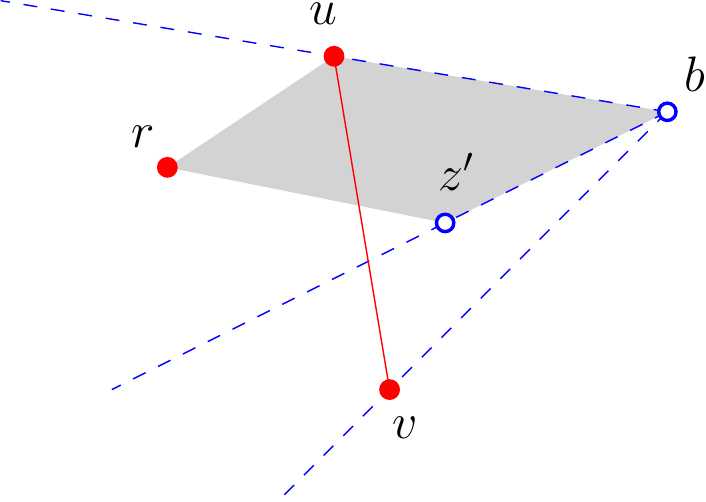}
		\label{fig:linear-sep-c}
	}\hspace{0.5cm}
	\subfloat[]{
		\includegraphics[scale=0.6]{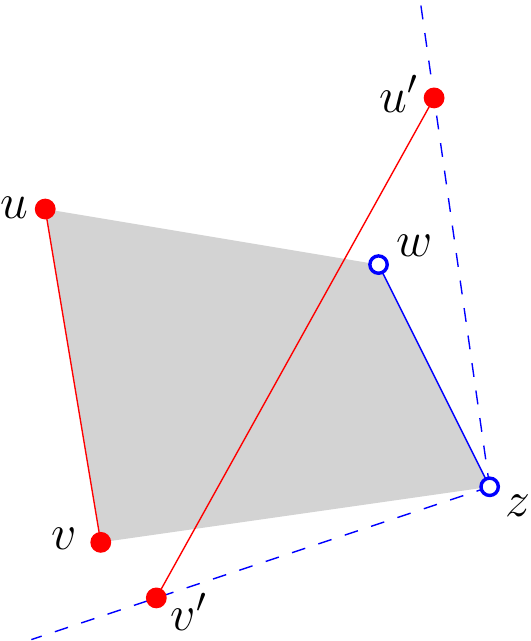}
		\label{fig:linear-sep-d}
	}
	\caption{\small{Proof of Theorem~\ref{thm:convex-4-hole-separable}.}}
	\label{fig:linear-sep}
\end{figure}

Suppose now that $S$ has a balanced convex $4$-hole with vertices 
$u,v,z,w$ in counter-clockwise order, where $u,v\in R$ and $w,z\in B$.
Let $e$ and $e'$ be the edges of $CH(R)$ and $CH(B)$, respectively, that intersect 
with both $\overline{uw}$ and $\overline{vz}$ (note that $e$ and $e'$
might share vertices with $\overline{uv}$ and $\overline{wz}$, respectively).
If we have that $e=\overline{uv}$ and $e'=\overline{wz}$ then
$e$ and $e'$ see each other, and thus {\em C1} holds.
Otherwise, if $e\neq\overline{uv}$ and $e'\neq\overline{wz}$ then
the interior of some edge among $\overline{uw}$, $\overline{uz}$, $\overline{vw}$, and $\overline{vz}$
intersects the interiors of both $e$ and $e'$.
Then, by Lemma~\ref{lem:red-blue-edge-crossed}, we have that
{\em C1} or {\em C2} holds.
Otherwise, there are two cases to consider: 
(1) $e\neq\overline{uv}$ and $e'=\overline{wz}$; and 
(2) $e=\overline{uv}$ and $e'\neq\overline{wz}$.
Consider case (1), case (2) is analogous. Let $e:=\overline{u'v'}$.
If $e$ and $e'$ see each other, then {\em C1} holds.
Otherwise (up to symmetry), $w$ belongs to $\Delta u'v'z$ (see Figure~\ref{fig:linear-sep-d}).
Since $R\cap \Delta u'v'z=\emptyset$ and $u\in\mathcal{W}(z,u',v')$,
we have that {\em C2} is satisfied. 
\end{proof}

\section{Discussion}\label{sec:discussion}

{\bf A better counting of black edges:}
In the proof of our lower bounds, we
considered the edges colored black, as those being
edges of the convex hull of $S=R\cup B$ ($|R|=|B|=n$) that
connect a red point with a blue point and are neither
an edge nor a diagonal of any balanced $4$-hole. Specifically, in the proof
of Lemma~\ref{lem:number-green-edges}, we gave the
simple upper bound $2n$ for the number of black edges,
but one can note that this bound can be improved.
Nevertheless, any upper bound must be at least $n/2$
since the following bicolored point set
has precisely $n/2$ black edges.

Let $n=4k$ and consider a regular $2k$-gon $Q$. Put a colored
point at each vertex of $Q$ such that
the colors of its vertices alternate along its boundary. Orient the edges of
$Q$ counter-clockwise. Then for each edge $e$ of $Q$
put in the interior of $Q$ three points of the color of 
the origin vertex of $e$
such that they are close enough to $e$ and ensure that there
is no balanced $4$-hole with $e$ as edge.
In total we have $8k$ points, consisting of $4k$ red points (i.e.\ $k$
red points in vertices of $Q$ and $3k$ red points in the interior of $Q$) and $4k$ blue points.
See for example Figure~\ref{fig:many-black-edges}, in which $k=2$.
Then, all the $2k=n/2$ edges of $Q$ are black.

\begin{figure}[h]
	\centering
	\includegraphics[scale=0.6]{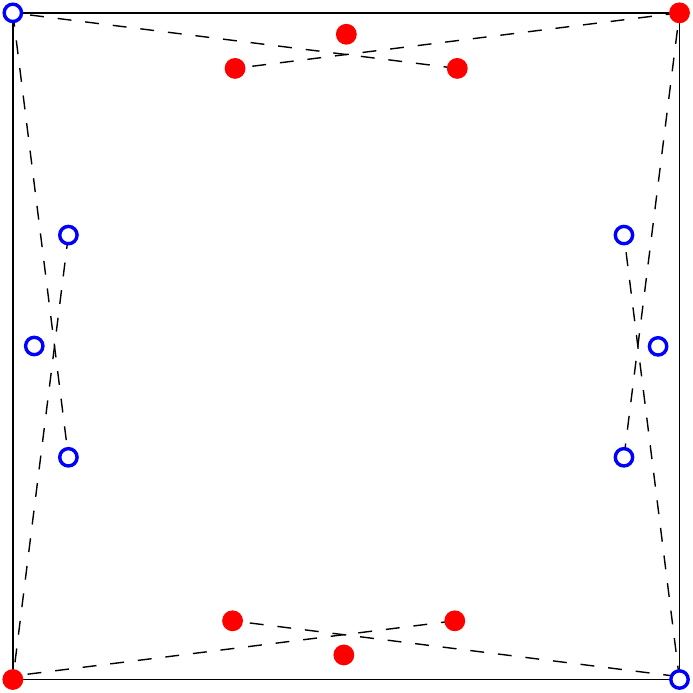}
	\caption{\small{A point set with many black edges.}}
	\label{fig:many-black-edges}
\end{figure}

\medskip

{\bf Generalization of the lower bound for non-balanced point sets:} 
Let $S=R\cup B$ be a red-blue colored point set such that $|R|\neq|B|$.
Let $\nr:=|R|$ and $\nb:=|B|$. Using
arguments similar to the ones used in Section~\ref{sec:lower-bound}, 
it can be proved that $S$ has at least
$$\nr\cdot\nb-\nr\cdot\min\left\{\left\lfloor\frac{\nr-1}{3}\right\rfloor,\nb\right\}
-\nb\cdot\min\left\{\left\lfloor\frac{\nb-1}{3}\right\rfloor,\nr\right\}-(\nr+\nb)$$
balanced $4$-holes.
Observe that this bound is positive if and only if 
$\nb$ is at least roughly $\nr(3-\sqrt{5})/2\approx 0.382\cdot\nr$ and
$\nr$ is at least $\nb(3-\sqrt{5})/2\approx 0.382\cdot\nb$.
Therefore,
we leave as an open problem to obtain a lower bound for the cases in which
the number of points of one color is below $0.382$ times the number 
of points of the other color.
%

\medskip

{\bf Existence of convex $4$-holes, either balanced or monochromatic:} 
Combining the characterization
given by Theorem~\ref{thm:convex-4-holes-not-separable} and by
Theorem~\ref{thm:convex-4-hole-separable}, we obtain the following
result:

\begin{proposition}
Let $S=R\cup B$ a bicolored point set in the plane. If $|R|,|B|\geq 4$
then $S$ always has a convex $4$-hole either balanced or monochromatic.
\end{proposition}

\begin{proof}
If $R$ and $B$ are not linearly separable, then $S$ has a
balanced convex $4$-hole by Theorem~\ref{thm:convex-4-holes-not-separable}. Otherwise,
consider that $R$ and $B$ are linearly separable. If the convex hull of
$R$ contains a red point and the convex hull of $B$ contains a blue point in their interiors, then
$S$ has a balanced convex $4$-hole by Lemma~\ref{lem:red-blue-edge-crossed}. Otherwise,
at least one between $R$ and $B$ is in convex position and then $S$ has a
monochromatic convex $4$-hole.
\end{proof}

\medskip

{\bf Deciding the existence of balanced convex $4$-holes:} 
Using the characterization Theorems~\ref{thm:convex-4-holes-not-separable} 
and~\ref{thm:convex-4-hole-separable}, arguments similar to those given
in Sections~\ref{sec:convex-4-holes:no-separability} and~\ref{sec:convex-4-holes:separability},
and well-known algorithmic results of computational geometry,
we can decide in $O(n\log n)$ time if a given bicolored point set $S=R\cup B$ ($|R|,|B|\ge 2$) of
total $n$ points has a balanced convex $4$-hole. 

We first 
compute the convex hulls $CH(R)$ and $CH(B)$ of $R$ and $B$, respectively.
After that, we decide if $R$ and $B$ are linearly separable. If they are not, 
we can decide in $O(n\log n)$ time whether one of the conditions (1-5) of Theorem~\ref{thm:convex-4-holes-not-separable}
holds. Otherwise, if $R$ and $B$ are linearly separable, we proceed with the following steps, each of them
in $O(n\log n)$ time.
If the decision performed in any of these steps has a positive answer, then a balanced convex $4$-hole exists:
\begin{enumerate}

\item Decide whether the next two conditions hold:
(1) $CH(R)$ contains red points in the interior or $CH(S)$ has
at least three red vertices; and 
(2) $CH(B)$ contains blue points in the interior or $CH(S)$ has
at least three blue vertices. 
If the answer is positive then the conditions of Lemma~\ref{lem:red-blue-edge-crossed}
are met and there thus exists a balanced convex $4$-hole in $S$.
Otherwise, if the answer is
negative, assume w.l.o.g.\ that $B$ is in convex position.

\item Decide whether the conditions of Lemma~\ref{lem:red-blue-edge-crossed} hold
for at least one red point $r$. Fixing a red point $r$, those conditions can be verified
in $O(\log n)$ time as follows: 
Let $b_0,b_1,\ldots,b_{m-1}$ be all the blue points labelled clockwise
along the boundary of $CH(B)$ (subindices are taken modulo $m$).
Let $b_i$ and $b_j$ be the two blue points such that $r\rightarrow b_i$ and $r\rightarrow b_j$
are tangent to $CH(B)$, and let $B_{i,j}:=\{b_{i+1},b_{i+2},\ldots,b_{j-1}\}$ be the points
between $b_i$ and $b_j$ which are on the side of the line $\ell(b_i,b_j)$ opposite
to the side containing $R$.
Note that the boundary of $CH(B)$ intersects the 
interior of $\overline{rb}$ for every $b\in B_{i,j}$.
If $r$ is a vertex of $CH(R)$, then it suffices to verify the existence of a blue point 
$b$ in $B_{i,j} \cap \mathcal{W}(r,r',r'')$,
where $r'$ and $r''$ are the vertices preceding and succeeding $r$, respectively, in
the boundary of $CH(R)$.
Otherwise, if $r$ belongs to the interior of $CH(R)$, then it suffices to verify the 
existence of a blue point $b$ in $B_{i,j}$, that is, $B_{i,j}$ is not empty.
Both $b_i$ and $b_j$ can be found in $O(\log n)$ time, as well the existence
of such a point $b$ can be decided in $O(\log m)=O(\log n)$ time by applying binary search over 
the points $b_{i+1},b_{i+2},\ldots,b_{j-1}$.


\item Decide whether Condition {\em C1} holds. This can be done in $O(n)$ time by simultaneously 
traversing the boundaries of $CH(R)$ and $CH(B)$.

\item Decide whether Condition {\em C2} holds. Using the fact that neither
condition {\em C1} nor the conditions of Lemma~\ref{lem:red-blue-edge-crossed} hold, 
we claim that condition {\em C2} can be decided by assuming that segment $\overline{bz}$ is an edge 
of $CH(B)$ and that point $z$ is the only blue point in the triangle $\Delta uvb$ (the condition
{\em C2} with $R$ and $B$ swapped is similar to decide).
Namely, let $\overline{uv}$ be an edge of $CH(R)$ and $b,z\in B$ be points such that
$z\in \Delta uvb$, $R\cap \Delta uvb=\emptyset$, and $R\cap\mathcal{W}(b,u,v)\neq\emptyset$. 
Let $z':=f(u,v,b,B)\neq b$, and
observe that at least one neighbor of $z'$ in the boundary of $CH(B)$, say $b'$, satisfies
$b'\in\Delta uvb\cup\{b\}$ and $z'$ is the only one blue point in $\Delta uvb'$.
The fact $R\cap\mathcal{W}(b,u,v)\subseteq R\cap\mathcal{W}(b',u,v)$ implies that
we can verify condition {\em C2} with $b'$ being $b$ and $z'$ being $z$, where $\overline{b'z'}$
is an edge of $CH(B)$ (see Figure~\ref{fig:deciding-a}). The claim thus follows.
Therefore, there is a linear-size set $W$
of wedges of the form $\mathcal{W}(b,u,v)$ to consider, and we need to check if there is an incidence between
any red point and an element of $W$. Note that the elements of $W$ can be divided into two
groups, such that in each group the intersections of the wedges with the interior of $CH(R)$
are pairwise disjoint (see Figure~\ref{fig:deciding-b}). The wedge $\mathcal{W}(b,u,v)$ goes to the first
group when $z$ is the clockwise neighbor of $b$ in the boundary of $CH(B)$, and to the other group otherwise.
Then, for each red point $r$, one can decide in $O(\log n)$ time
such an incidence.
\end{enumerate}

\begin{figure}[h]
	\centering
	\subfloat[]{
		\includegraphics[scale=0.6,page=1]{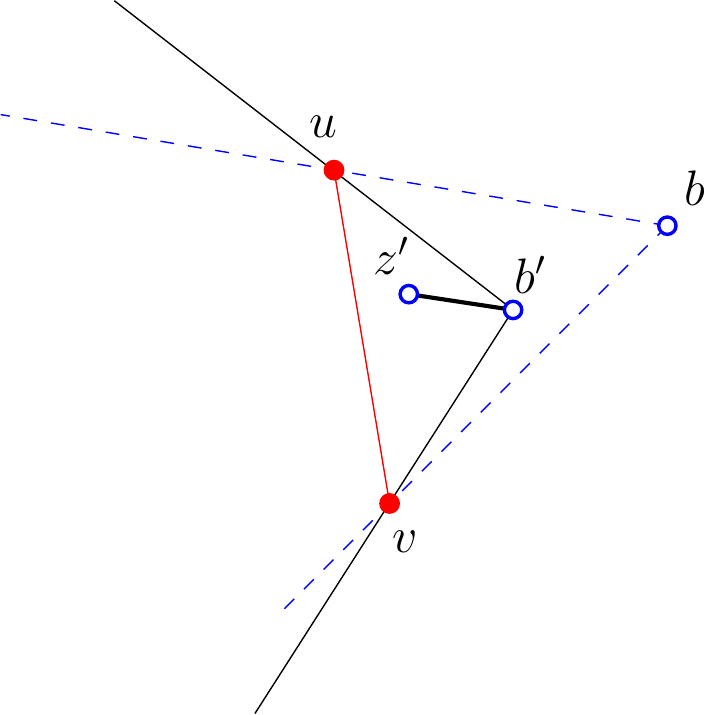}
		\label{fig:deciding-a}
	}\hspace{0.5cm}
	\subfloat[]{
		\includegraphics[scale=0.6,page=2]{img/decide-convex-4hole.pdf}
		\label{fig:deciding-b}
	}
	\caption{\small{Deciding the existence of a balanced convex $4$-hole.
	}}
	\label{fig:deciding}
\end{figure}

\medskip

{\bf Counting balanced $4$-holes:}
Adapting the algorithm of Mitchell et al.~\cite{mitchellRSW95} for counting convex polygons in planar
point sets, we can count the balanced $4$-holes of a bicolored
point set $S$ of $n$ points in $O(\tau(n))$ time,
where $\tau(n)$ is the number of empty triangles of $S$.

\medskip

{\bf Existence of balanced $2k$-holes in balanced point sets:} 
The arguments used to prove the existence of at least one balanced
$4$-hole in any point set $S=R\cup B$ with $|R|,|B|\ge 2$ 
(at the beginning of Section~\ref{sec:lower-bound})
do not directly apply to prove the existence of balanced $2k$-holes in
point sets $S=R\cup B$ with $|R|,|B|\ge k$. However, we can 
prove the following:

\begin{proposition}
For all $n\ge 1$ and $k\in[1..n]$, every point set $S=R\cup B$ with $|R|=|B|=n$
contains a balanced $2k$-hole.
\end{proposition}

\begin{proof}
If $S$ is in convex position then the result follows. Then, suppose that 
$S$ is not in convex position.
For every point $p\in R$ let $w(p):=1$, and for every $p\in B$ let $w(p):=-1$.
W.l.o.g.\ let $u\in B$ be a point in the interior of $CH(S)$, and $p_0,p_1,\ldots,p_{2n-2}$
denote the elements of $S\setminus\{u\}$ sorted radially in clockwise order around $u$.
For $i=0,1,\ldots,2n-2$, let $s_i:=w(p_i)+w(p_{i+1})+\ldots+w(p_{i+2k-2})$, where
subindices are taken modulo $2n-1$. Notice that all $s_i$'s are odd, and 
$s_i=1$ implies that the points 
$u,p_i,p_{i+1},\ldots,p_{i+2k-2}$ form a balanced $2k$-hole.

We have that $\sum_{i=0}^{2n-2}s_i=(2k-1)\sum_{i=0}^{2n-2}w(p_i)=2k-1$, which implies 
(given that $k\in[1..n]$)
that not all $s_i$'s can be greater than $1$ and that not all $s_i$'s can be smaller than $1$.
Suppose for the sake of contradiction that none of the $s_i$'s is equal to $1$. 
Then, there exist an $s_j<1$ and an $s_t>1$.
Since we further have that $s_i-s_{i+1}\in\{-2,0,2\}$ for all $i\in[0..2n-2]$, there must exist an element
among $s_{j+1},s_{j+2},\ldots,s_{t-1}$ which is equal to $1$, and the result thus follows.
\end{proof}

\medskip

{\bf Open problems:} As mentioned above,
we leave as open the problem of obtaining a lower bound
for the number of balanced $4$-holes in point sets $S=R\cup B$
in which either $|R|<0.382\cdot|B|$ or $|B|<0.382\cdot|R|$.
Another open problem
is to study lower bounds on the number of balanced $k$-holes,
for even $k\geq 6$. 

\section*{Acknowledgements}

The authors would like to thank the two
anonymous referees for their useful comments and suggestions.

\small

\bibliographystyle{plain}
\bibliography{4-holes}

\begin{thebibliography}{10}

\bibitem{AichUrr}
O.~Aichholzer, R.~Fabila-Monroy, H.~Gonz\'{a}lez-Aguilar, T.~Hackl, M.~A.
  Heredia, C.~Huemer, J.~Urrutia, and B.~Vogtenhuber.
\newblock 4-holes in point sets.
\newblock In {\em 27th European Workshop on Computational Geometry EuroCG'11},
  pages 115--118, Morschach, Switzerland, 2011.

\bibitem{aichholzer2010}
O.~Aichholzer, T.~Hackl, C.~Huemer, F.~Hurtado, and B.~Vogtenhuber.
\newblock Large bichromatic point sets admit empty monochromatic 4-gons.
\newblock {\em SIAM J. Discret. Math.}, 23(4):2147--2155, January 2010.

\bibitem{AichUrrVirg2013}
O.~Aichholzer, J.~Urrutia, and B.~Vogtenhuber.
\newblock Balanced 6-holes in linearly separable bichromatic point sets.
\newblock {\em Electronic Notes in Discrete Mathematics}, 44(0):181--186, 2013.

\bibitem{bdfprauv-bh-15}
S.~Bereg, J.~M. D\'{\i}az-B{\'a}{\~n}ez, R.~Fabila-Monroy, P.~Perez-Lantero,
  A.~Ramirez-Vigueras, T.~Sakai, J.~Urrutia, and I.~Ventura.
\newblock On balanced 4-holes in bichromatic point sets.
\newblock {\em Comput. Geom. Theory Appl.}, 48(3):169--179, 2015.

\bibitem{devillers}
O.~Devillers, F.~Hurtado, G.~K\'{a}rolyi, and C.~Seara.
\newblock Chromatic variants of the {E}rd{\H o}s-{S}zekeres theorem on points
  in convex position.
\newblock {\em Comput. Geom. Theory Appl.}, 26(3):193--208, November 2003.

\bibitem{Erdos1}
P.~Erd{\H o}s.
\newblock Some more problems on elementary geometry.
\newblock {\em Austral. Math. Soc. Gaz.}, 5:52--54, 1978.

\bibitem{Erdos2}
P.~Erd{\H o}s and G.~Szekeres.
\newblock A combinatorial problem in geometry.
\newblock {\em Compositio Math.}, 2:463--470, 1935.

\bibitem{garciaNT00}
A.~Garc\'{\i}a, M.~Noy, and J.~Tejel.
\newblock Lower bounds on the number of crossing-free subgraphs of
  k$_{\mbox{n}}$.
\newblock {\em Comput. Geom. Theory Appl.}, 16(4):211--221, 2000.

\bibitem{GERK}
T.~Gerken.
\newblock Empty convex hexagons in planar point sets.
\newblock {\em Discrete Comput. Geom.}, 39(1):239--272, March 2008.

\bibitem{HARB}
H.~Harborth.
\newblock Konvexe f\"unfecke in ebenen punktmengen.
\newblock {\em Elem. Math.}, 33:116--118, 1978.

\bibitem{Hort}
J.~D. Horton.
\newblock Sets with no empty convex 7-gons.
\newblock {\em Canad. Math. Bull.}, 26:482--484, 1983.

\bibitem{clemens}
C.~Huemer and C.~Seara.
\newblock 36 two-colored points with no empty monochromatic convex fourgons.
\newblock {\em Geombinatorics}, XIX(1):5--6, 2009.

\bibitem{koshelev}
V.~Koshelev.
\newblock On {E}rd{\H o}s-{S}zekeres problem and related problems.
\newblock ArXiv e-prints, 2009.

\bibitem{mitchellRSW95}
J.~S.~B. Mitchell, G.~Rote, G.~Sundaram, and G.~J. Woeginger.
\newblock Counting convex polygons in planar point sets.
\newblock {\em Inf. Process. Lett.}, 56(1):45--49, 1995.

\bibitem{CNico}
C.~M. Nicolas.
\newblock The empty hexagon theorem.
\newblock {\em Discrete Comput. Geom.}, 38(2):389--397, September 2007.

\end{thebibliography}

\end{document}